\documentclass[12pt]{amsproc}

\usepackage[left=1in, right=1in, top=1in, bottom=1in]{geometry}
\usepackage[authoryear]{natbib}
\usepackage{amsmath,amsthm,amsfonts,amssymb,bm}
\usepackage[T1]{fontenc}
\usepackage{mathtools}
\usepackage[dvipsnames]{xcolor}
\usepackage[colorlinks=true,citecolor=blue,urlcolor=magenta,linkcolor=blue]{hyperref}
\usepackage{blkarray}
\usepackage{tikz,tikz-cd}
\usepackage{cleveref}
\usepackage{bbm}
\usepackage{algorithm}
\usepackage{algpseudocode}
\usepackage{enumitem}
\usepackage{subcaption}
\usepackage{graphicx}
\mathtoolsset{showonlyrefs}

\usetikzlibrary{chains,matrix,scopes,fit,decorations,positioning,shapes,arrows,decorations.markings,fit,backgrounds}
\newcommand\cD{\mathcal{D}}

\newcommand\cG{\mathcal{G}}

\newcommand\cI{\mathcal{I}}

\newcommand\cS{\mathcal{S}}
\newcommand\cT{\mathcal{T}}

\newcommand\T{\mathcal{T}}

\newcommand{\cum}{{\rm cum}}
\newcommand{\pa}{\mathrm{pa}}

\newcommand\E{\mathbb{E}}

\newcommand\R{\mathbb{R}}
\newcommand\cV{\mathcal{V}}
\newcommand\Sym{\mathrm{Sym}}

\DeclareMathOperator*{\argmin}{arg\,min}

\DeclareMathOperator{\Cov}{Cov}
\DeclareMathOperator{\Var}{Var}

\theoremstyle{plain}
\newtheorem{thm}{thm}[section]
\newtheorem{theorem}[thm]{Theorem}
\newtheorem{proposition}[thm]{Proposition}
\newtheorem{lemma}[thm]{Lemma}

\theoremstyle{definition}

\newtheorem{example}[thm]{Example}
\newtheorem{assumption}[thm]{Assumption}

\theoremstyle{remark}
\newtheorem{remark}[thm]{Remark}

\AddToHook{env/theorem/begin}{\crefalias{thm}{theorem}}
\AddToHook{env/proposition/begin}{\crefalias{thm}{proposition}}
\AddToHook{env/lemma/begin}{\crefalias{thm}{lemma}}
\AddToHook{env/corollary/begin}{\crefalias{thm}{corollary}}
\AddToHook{env/conjecture/begin}{\crefalias{thm}{conjecture}}
\AddToHook{env/definition/begin}{\crefalias{thm}{definition}}
\AddToHook{env/example/begin}{\crefalias{thm}{example}}
\AddToHook{env/assumption/begin}{\crefalias{thm}{assumption}}
\AddToHook{env/remark/begin}{\crefalias{thm}{remark}}

\crefname{thm}{theorem}{theorems}
\Crefname{thm}{Theorem}{Theorems}
\crefname{theorem}{theorem}{theorems}
\Crefname{theorem}{Theorem}{Theorems}
\crefname{proposition}{proposition}{propositions}
\Crefname{proposition}{Proposition}{Propositions}
\crefname{lemma}{lemma}{lemmas}
\Crefname{lemma}{Lemma}{Lemmas}
\crefname{corollary}{corollary}{corollaries}
\Crefname{corollary}{Corollary}{Corollaries}
\crefname{conjecture}{conjecture}{conjectures}
\Crefname{conjecture}{Conjecture}{Conjectures}
\crefname{definition}{definition}{definitions}
\Crefname{definition}{Definition}{Definitions}
\crefname{example}{example}{examples}
\Crefname{example}{Example}{Examples}
\crefname{assumption}{assumption}{assumptions}
\Crefname{assumption}{Assumption}{Assumptions}
\crefname{remark}{remark}{remarks}
\Crefname{remark}{Remark}{Remarks}

\title{Causal discovery under mean independence and linearity}

\author{Geert Mesters}
\address{Federal Reserve Bank of New York}
\email{geert.mesters@ny.frb.org}

\author{Alvaro Ribot}
\address{Harvard University}
\email{aribotbarrado@g.harvard.edu}

\author{Anna Seigal}
\address{Harvard University}
\email{aseigal@seas.harvard.edu}

\author{Piotr Zwiernik}
\address{Universitat Pompeu Fabra and Barcelona School of Economics}
\email{piotr.zwiernik@upf.edu}

\begin{document}

\begin{abstract}
Causal discovery methods such as LiNGAM identify causal structure from observational data by assuming mutually independent disturbances. This assumption is fragile: shared volatility, common scale effects, or other forms of dependence can cause the methods to recover the wrong causal order, even with infinite data. We introduce the Linear Mean-Independent Acyclic Model (LiMIAM), which replaces full independence with weaker one-sided mean-independence restrictions on the disturbances. 
Under finite-order consequences of these restrictions, source nodes are generically identifiable, and hence a compatible causal order can be recovered recursively. Our proof is constructive and leads to DirectLiMIAM, a sequential residual-based algorithm for causal discovery under dependent noise. In simulations with mean-independent but dependent disturbances, DirectLiMIAM outperforms LiNGAM methods.
A large-scale empirical application to the oil market highlights the implausibility of the independence assumption and the ability of DirectLiMIAM to recover a realistic causal ordering, from policy to production and from prices to inflation.
\end{abstract}

\maketitle

\begingroup
\renewcommand\thefootnote{}
\footnotetext{All authors contributed equally.}
\endgroup

\section{Introduction}

Recovering causal structure from observational data is a central problem in statistics, econometrics, and many other empirical sciences. Linear structural equation models (LSEMs) are a natural framework for representing causal mechanisms, but they are not identifiable from the joint distribution alone: without additional assumptions, different directed acyclic graphs (DAGs) can induce the same statistical model \citep{Dhrymes1994,pearl2009causality, spirtes2000causation}. A major advance in this area was the Linear Non-Gaussian Acyclic Model (LiNGAM) of \citet{Shimizu2006LiNGAM}. Under the assumption that the causal mechanisms are linear and the structural disturbances are mutually independent and non-Gaussian, LiNGAM achieves identifiability of the causal DAG, by connecting causal discovery with independent component analysis (ICA) \citep{Comon1994}. It has led to an extensive literature, including direct estimation \citep{Shimizu2011DirectLiNGAM}, likelihood-based alternatives \citep{Hyvarinen2013}, high-dimensional methods based on higher-order moments \citep{WangDrton2020HighDimensionalLiNGAM}, extensions to time-series and structural vector autoregressive models (SVARs) \citep{LanneLutkepohl2010,Hyvarinen2010,Lanne2017,GourierouxMonfortRenne2017,Hoesch2024}, and treatments of latent confounding \citep{Hoyer2008latent,Wang2023causal}. At the level of higher-order statistics, independence implies diagonal cumulant tensors, a viewpoint that is explicit in tensor-diagonalization formulations of ICA.

The key assumption underlying LiNGAM is that the errors are mutually independent and non-Gaussian. While non-Gaussianity is often plausible, mutual independence is harder to justify. A more realistic assumption is that the disturbances satisfy
weaker first-moment restrictions while still exhibiting dependence in their higher moments. Such dependence is common. In economics and finance, disturbances often exhibit conditional heteroskedasticity, stochastic volatility, and common scale
factors \citep{Engle1982,Bollerslev1986,Harvey1994}. In SVARs, non-Gaussian identification has been proposed as
an alternative to classical exclusion restrictions \citep{LanneLutkepohl2010,Hyvarinen2010,Lanne2017,GourierouxMonfortRenne2017}, but dependent non-Gaussian disturbances arise through volatility clustering and
scale mixtures. In single-cell biology, cell-to-cell variability often
contains extrinsic components, driven, for example, by differences in cell state, cell cycle, or other factors shared across genes or pathways
\citep{Elowitz2002,Brennecke2013,Buettner2015,Stegle2015}. In neuroscience, ongoing activity and common state fluctuations can induce shared variability across neural responses \citep{Arieli1996,CohenKohn2011,Rabinowitz2015}. In such situations, it
may remain plausible that the mean of one disturbance cannot
be predicted from earlier disturbances even though full independence
fails. These exogeneity-type conditions motivate the framework in this paper.

If independence is assumed when it is false, LiNGAM procedures can produce  incorrect causal conclusions. We show that standard independence-based
implementations, including JADE \citep{cardoso1996jacobi}, FastICA
\citep{hyvarinen1999fast}, and residual-based LiNGAM procedures
\citep{Shimizu2011DirectLiNGAM}, can recover the wrong causal order
even at the population level. Under disturbances that
satisfy mean-independence conditions but are dependent, these methods may reverse the
direction of causation. We complement the population counterexamples
with simulations showing that the problem becomes more severe as the
dimension grows: under dependent but mean-independent noise, the graph
recovery performance of LiNGAM methods deteriorates, whereas our method remains stable.

The starting point of this paper is that the assumptions needed
for causal identification are weaker than those of
independence-based LiNGAM procedures. We introduce the \emph{Linear
Mean-Independent Acyclic Model} (LiMIAM), in which the
disturbances may be dependent but must satisfy one-sided mean-independence
conditions, which we call order-dependent mean independence. The underlying
causal object is a DAG. Whenever we write matrices or tensors in triangular
form, we have fixed a compatible topological order and relabeled the
variables accordingly. In this coordinate system, whenever one disturbance
appears later than another in the chosen topological order, the later
disturbance cannot be predicted in mean from the earlier one. This allows higher-order dependence
and is weaker than usual independence-based assumptions, including the
diagonal cumulant structure that underlies ICA, see~\Cref{fig:schematic}.
Our main result is that, under finite-order moment conditions implied by
this order-dependent mean independence, source nodes can be identified
generically, and hence a compatible causal order is recovered recursively.
Once an order is known, the coefficient matrix, and hence the DAG, can be
estimated by regression restricted to earlier variables in the order. By
generically, we mean away from a set of measure zero in the relevant
parameter space.

\begin{figure}[htbp]
\centering
\begin{tikzpicture}[
    font=\normalsize,
    head/.style={font=\bfseries\color{black!80}, align=center},
    arrow/.style={font=\large\color{black!65}},
    relaxarrow/.style={font=\Large\color{orange!85!black}},
    rowlabel/.style={font=\small\bfseries},
    rowbox/.style={
        rounded corners,
        minimum width=11.4cm,
        minimum height=1.55cm,
        line width=0.8pt
    },
    oldbox/.style={rowbox, draw=blue!60!black, fill=blue!7},
    midbox/.style={rowbox, draw=teal!65!black, fill=teal!7},
    newbox/.style={rowbox, draw=orange!85!black, fill=yellow!22, line width=1.2pt},
    lefttext/.style={align=center, text width=3.75cm},
    righttext/.style={align=center, text width=5.05cm}
]

\def\xL{0}
\def\xM{4.05}
\def\xR{7.55}
\def\xC{3.80}

\node[head, text width=3.9cm] at (\xL,0.35)
    {\rm distribution of disturbances};
\node[head, text width=5.1cm] at (\xR,0.35)
    {\rm moments/cumulants of disturbances};

\node[oldbox] at (\xC,-1.0) {};
\node[midbox] at (\xC,-3.15) {};
\node[newbox] at (\xC,-5.30) {};
\node[rowlabel, text=blue!70!black, anchor=east]
    at (-2.05,-1.0) {LiNGAM};
\node[rowlabel, text=orange!85!black, anchor=east]
    at (-2.05,-5.30) {LiMIAM};

\node[lefttext] at (\xL,-1.0) {independent};
\node[arrow] at (\xM,-1.0) {$\Longrightarrow$};
\node[righttext] at (\xR,-1.0) {diagonal cumulants};

\node[relaxarrow] at (\xL,-2.05) {$\Downarrow$};
\node[relaxarrow] at (\xR,-2.05) {$\Downarrow$};

\node[lefttext] at (\xL,-3.15) {pairwise\\mean independent};
\node[arrow] at (\xM,-3.15) {$\Longrightarrow$};

\node[righttext] at (\xR,-3.15)
{
$(i,\ldots,i,j)$ entry is $0$\\
for all $i\neq j$
};

\node[relaxarrow] at (\xL,-4.25) {$\Downarrow$};
\node[relaxarrow] at (\xR,-4.25) {$\Downarrow$};

\node[lefttext] at (\xL,-5.30)
{
order-dependent\\
mean independent
};

\node[arrow] at (\xM,-5.30) {$\Longrightarrow$};

\node[righttext] at (\xR,-5.30)
{
$(i,\ldots,i,j)$ entry is $0$\\
when $i$ precedes $j$ in the causal order
};

\node[
    font=\Huge,
    color=orange!80!black
] at (-2.85,-3.15)
{$\Downarrow$};
\end{tikzpicture}
\caption{From LiNGAM to LiMIAM.}
\label{fig:schematic}
\end{figure}
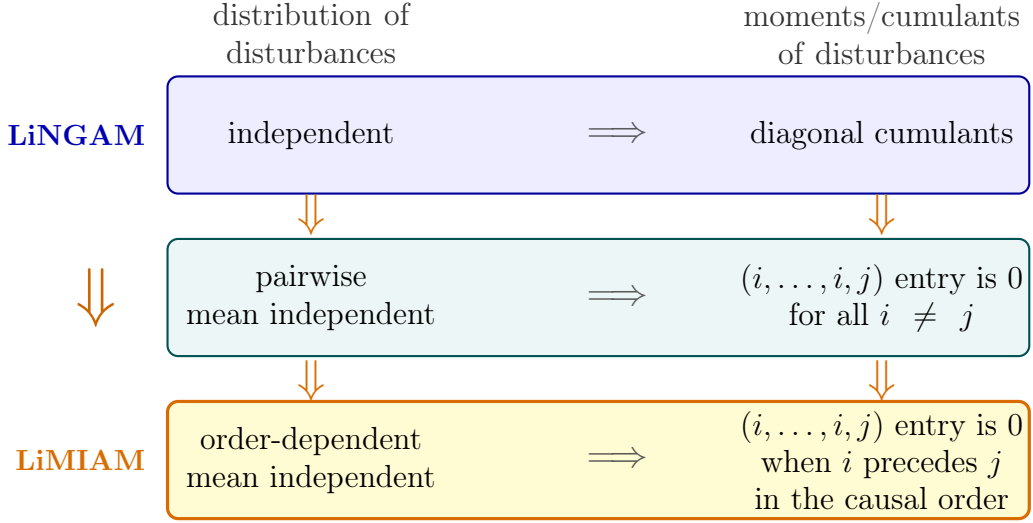

Our work is related in spirit to recent developments in
non-independent component analysis
\citep{MestersZwiernik2024NICA,Ribot2025TensorEigenvectors,RibotSeigalZwiernik2025BeyondICA}, which
show that the independence assumption in ICA can be relaxed. The present
paper goes beyond this work in two ways. First, it uses
that the recursive causal model
mixing matrix is lower-triangular, after a permutation,
which allows weaker assumptions than those needed for
identifiability of an unrestricted mixing. Second, it develops a
causal-discovery algorithm based on recursive identification of source
variables and residualization. This is a LiNGAM-type algorithm not driven by ICA, which extends the logic of
DirectLiNGAM \citep{Shimizu2011DirectLiNGAM} to
dependent disturbances. Accordingly, the paper contributes at three levels: failure analysis,
identification theory, and methodology. For finite-moments,
our viewpoint leads to an algebraic model for coupled
moments under order-dependent mean independence. More broadly, it shows that
structural identifiability of the causal model is a separate issue from
the validity of specific independence-based LiNGAM procedures.

The finite-moment identification argument also connects to the
higher-order moment criterion of \citet{WangDrton2020HighDimensionalLiNGAM}.
They study high-dimensional LiNGAM under independent non-Gaussian
disturbances and formulate their conditions in terms of DAG
parents. We use the same algebraic type of moment identity under weaker
one-sided moment restrictions, where disturbances may be dependent. The
resulting non-cancellation condition is the LiMIAM analogue of their
parental faithfulness condition.

The rest of the paper is organized as follows. In~\Cref{sec:model}, we introduce the recursive linear model, the order-dependent mean independence assumption and show the failure of LiNGAM procedures under dependent but mean-independent disturbances. In~\Cref{sec:ident}, we show that source nodes, and hence a compatible causal order, can be recovered generically from a pair of moments under order-dependent mean independence, and we study a distribution-level route based on residual mean independence. In~\Cref{sec:algorithm}, we describe the DirectLiMIAM algorithm and its relation to existing methods. \Cref{sec:synthetic} presents numerical experiments on simulated data that show the performance of DirectLiMIAM. \Cref{sec:benchmarks} presents an empirical application for the oil market where we connect to the economic literature on SVAR models and show that DirectLiMIAM can identify an economically interpretable DAG, with causal order from policy to production and prices to inflation.

Our Python code with an implementation of \Cref{alg:DirectLiMIAM} and scripts to replicate the numerical experiments is available on GitHub: \href{https://github.com/alvaro-ribot/limiam}{https://github.com/alvaro-ribot/limiam}.

\section{Motivation: non-independence and the brittleness of LiNGAM}
\label{sec:model}

\subsection{Recursive linear models and first-moment assumptions}

We consider a linear structural equation model (LSEM) for random variables
\(X_1,\dots,X_p\). The underlying causal structure is a directed
acyclic graph (DAG). For notational convenience, we fix a compatible
topological order and relabel variables so that parents precede children,
written \(1<\cdots<p\). Thus,
\begin{equation}
\label{eq:lmiam-coordinate}
X_i=\sum_{j<i}\beta_{ij}X_j+\varepsilon_i,
\qquad i=1,\dots,p,
\end{equation}
where some coefficients \(\beta_{ij} \in \mathbb{R}\) may be zero.
The variable \(\varepsilon_i\) is the disturbance for \(X_i\). The set of parents
of \(i\), denoted $\pa(i)$, are the indices \(j<i\) for which \(\beta_{ij}\neq 0\), and
these parent sets define the DAG. See \Cref{fig:dag-lsem-example}. A variable is \emph{exogenous}, or a
\emph{source}, if it has no parents.

Writing \(B=(\beta_{ij})_{ij}\in\R^{p\times p}\),
the model \eqref{eq:lmiam-coordinate} is equivalent to
\begin{equation}
\label{eq:sem}
X=BX+\varepsilon,
\end{equation}
where \(B\) is strictly lower triangular. Hence
\begin{equation}
\label{eq:mix}
X=A\varepsilon,
\qquad \text{where} \qquad 
A:=(I-B)^{-1},
\end{equation}
and \(A\) is lower triangular with unit diagonal.

A standard structural exogeneity condition is that each disturbance has
conditional mean zero given the variables on the right-hand side of its
equation in \eqref{eq:lmiam-coordinate}, that is
\begin{equation}
\label{eq:parent-mean-zero}
\E[\varepsilon_i\mid X_{\pa(i)}]=0,
\qquad i=1,\dots,p,
\end{equation}
which says that the unexplained part of \(X_i\) cannot be predicted in mean
from its direct causes. In LiMIAM we impose related first-moment
restrictions on the disturbances. These restrictions still allow dependence
in higher moments. See \Cref{fig:samples-example}.

\begin{figure}[t]
\centering
\begin{subfigure}[b]{0.3\textwidth}
\centering
\begin{minipage}[c][4cm][c]{\linewidth}
\centering
\[
\begin{cases}
X_1 = \varepsilon_1 \\
X_2 = \beta_{21} X_1 + \varepsilon_2 \\
X_3 = \beta_{31} X_1 + \varepsilon_3 \\
X_4 = \beta_{42} X_2 + \beta_{43} X_3 + \varepsilon_4
\end{cases}
\]
\end{minipage}
\caption{LSEM} \label{fig:lsem-example}
\end{subfigure}
\begin{subfigure}[b]{0.3\textwidth}
\centering
\begin{minipage}[c][4cm][c]{\linewidth}
\centering
\begin{tikzpicture}[
    node distance=1.5cm,
    every node/.style={circle, draw, minimum size=1cm}
]
\node (X1) {$X_1$};
\node (X2) [below left of=X1] {$X_2$};
\node (X3) [below right of=X1] {$X_3$};
\node (X4) [below right of=X2] {$X_4$};
\draw[->, thick] (X1) -- (X2);
\draw[->, thick] (X1) -- (X3);
\draw[->, thick] (X2) -- (X4);
\draw[->, thick] (X3) -- (X4);
\end{tikzpicture}
\end{minipage}
\caption{DAG} \label{fig:dag-example}
\end{subfigure}
\begin{subfigure}[b]{0.3\textwidth}
\centering
\begin{minipage}[c][4cm][c]{\linewidth}
\centering
\includegraphics[width=0.8\linewidth]{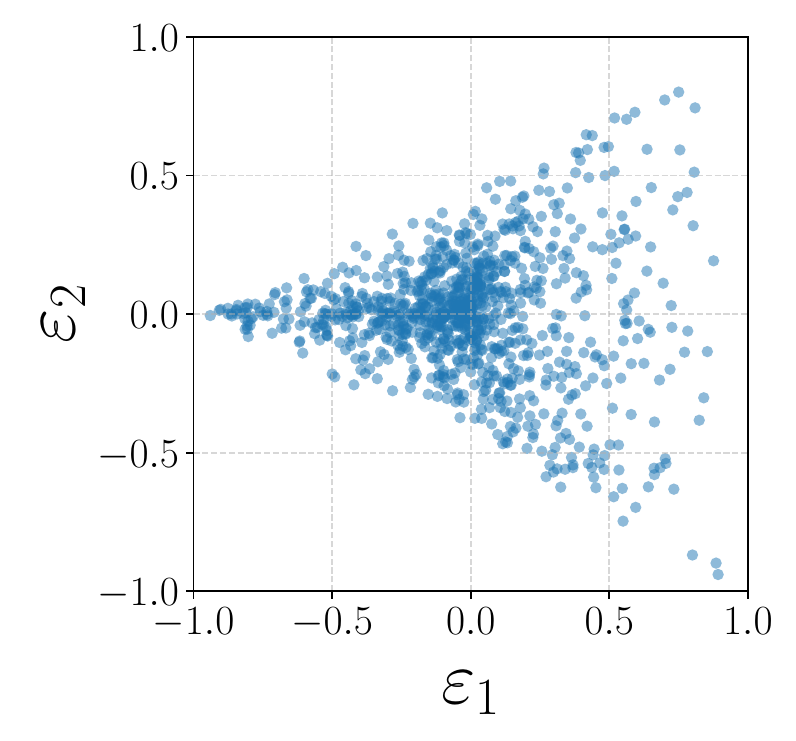}
\end{minipage}
\caption{Mean independence} \label{fig:samples-example}
\end{subfigure}

\caption{
(\textsc{a}) Linear structural equation model, (\textsc{b}) corresponding directed acyclic graph, and (\textsc{c}) order-dependent mean independent disturbances, i.e., $\E[\varepsilon_2 \mid \varepsilon_1] = 0$ but $\E[\varepsilon_1 \mid \varepsilon_2] \neq 0$.}
\label{fig:dag-lsem-example}

\end{figure}
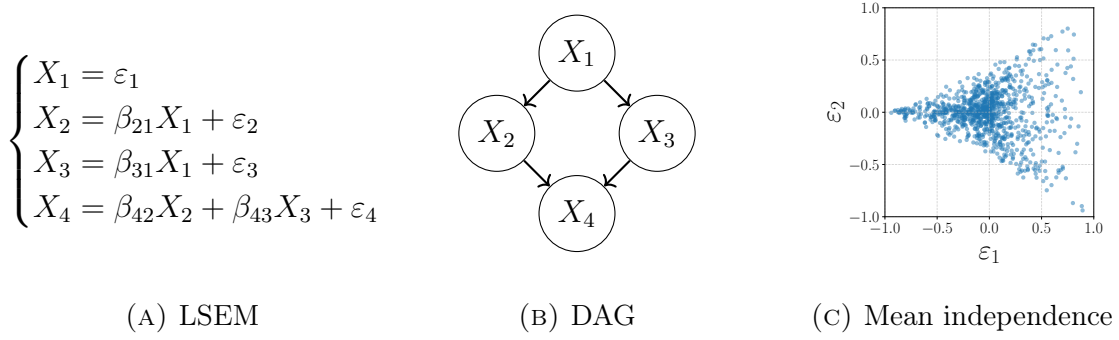

LiNGAM strengthens such first-moment exogeneity to mutual independence of
the coordinates of \(\varepsilon\). We show that identifiability still
holds under weaker first-moment conditions, which rule out certain
one-sided mean predictions among disturbances while allowing dependence in
higher moments.
The motivation is that independence-based
LiNGAM procedures can fail when mutual independence is replaced by
weaker first-moment assumptions. This is not a
small-sample issue: even at the population level, LiNGAM
methods may recover the wrong causal order.

\subsection{Population failure of independence-based LiNGAM}
\label{sec:fragility}
\label{sec:fragility-population}

We recall the classical LiNGAM setup
\citep{Shimizu2006LiNGAM}. In the model
\eqref{eq:sem}--\eqref{eq:mix}, if the structural disturbances
\(\varepsilon=(\varepsilon_1,\dots,\varepsilon_p)\) are mutually
independent and at most one component is Gaussian, then the mixing
matrix \(A\) is identifiable up to permutation and scaling via ICA
\citep{Comon1994}. The ICA--LiNGAM algorithm first estimates an
unmixing matrix \(W=A^{-1}\) using ICA and then applies row permutation
and scaling to obtain a strictly lower-triangular representation of~\(B\).

We first analyze ICA-based LiNGAM under weaker first-moment
assumptions. Assume \(\E[\varepsilon]=0\) and \(\Cov(\varepsilon)=I\) for simplicity. Procedures such as JADE \citep{cardoso1996jacobi} are
justified by mutual independence of the structural errors. Under mutual
independence, the fourth-order cumulant tensor of \(\varepsilon\) is
diagonal. Under weaker mean-independence assumptions, this is no longer
true (cf. Figure~\ref{fig:schematic}). For example, the joint cumulant
\[
\cum(\varepsilon_i,\varepsilon_i,\varepsilon_j,\varepsilon_j)
=
\Cov(\varepsilon_i^2,\varepsilon_j^2)
\]
need not vanish.

We illustrate the failure mechanism in the bivariate model
\[
X_1=\varepsilon_1,
\qquad
X_2=X_1+\varepsilon_2,
\]
which corresponds to the true DAG \(X_1\to X_2\) and mixing matrix
\begin{equation}
\label{eq:trueA}
A=
\begin{pmatrix}
1&0\\
1&1
\end{pmatrix}.
\end{equation}
Let
\(\kappa_4(\varepsilon)\) be the fourth-order cumulant tensor of
\(\varepsilon\) with
\[
\kappa_4(\varepsilon)_{1111}=k_1,\qquad
\kappa_4(\varepsilon)_{2222}=k_2,\qquad
\kappa_4(\varepsilon)_{1122}=c,\qquad
\kappa_4(\varepsilon)_{1112}
=
\kappa_4(\varepsilon)_{1222}
=
0,
\]
where \(k_1,k_2>0\) and \(c>0\). We argue that LiNGAM-based methods recover the wrong causal order if $c$ is large enough.

\begin{theorem}
\label{thm:lingam_reversal}
If $(k_1+k_2+6c)^2 < 8(k_1^2+k_2^2)$, then the population JADE step
recovers the true mixing matrix in \eqref{eq:trueA}. If the
reverse inequality holds, it recovers
\[
\hat A=
\frac{1}{\sqrt2}
\begin{pmatrix}
1 & 1\\
0 & 2
\end{pmatrix},
\]
which leads to the reversed causal order \(X_2\to X_1\). In the
boundary case of equality, both mixing matrices are global maximizers.
\end{theorem}

\Cref{thm:lingam_reversal} shows that, even with infinite data, JADE-based ICA--LiNGAM can recover the
wrong causal order. Its proof is in Appendix~\ref{sec:proofs-appendix}. Here,
strong dependence in the squared disturbances, as measured
by \(c\), causes the procedure to select the reversed DAG.

\begin{example}
\label{ex:common-variance-reversal}
The parameter region in \Cref{thm:lingam_reversal} is nonempty, and it
arises under common-variance disturbances. Let
\[
\varepsilon_1=\sigma Z_1,
\qquad
\varepsilon_2=\sigma Z_2,
\]
where \(Z_1,Z_2\) are independent copies of a uniform
random variable on \([-\sqrt3,\sqrt3]\), and \(\sigma\) is independent
of \((Z_1,Z_2)\). Assume that \(\sigma^2\) takes values \(0.1\) and
\(1.9\) with probability \(1/2\) each. Then \(\E[\sigma^2]=1\) and
\(\E[\sigma^4]=(0.1^2+1.9^2)/2=1.81\), so \(\Cov(\varepsilon)=I\).
Since \(Z_1\) and \(Z_2\) are symmetric and independent, we have
\(\kappa_4(\varepsilon)_{1112}
=
\kappa_4(\varepsilon)_{1222}
=
0\). Moreover,
\[
c=\kappa_4(\varepsilon)_{1122}
=
\Cov(\varepsilon_1^2,\varepsilon_2^2)
=
\E[\sigma^4]-1
=
0.81.
\]
Because \(\E[Z_1^4]=9/5\), we also obtain
\[
k_1=k_2
=
\kappa_4(\varepsilon_1)
=
\E[\sigma^4]\E[Z_1^4]-3
=
1.81\cdot 1.8-3
=
0.258.
\]
Hence
\[
(k_1+k_2+6c)^2
=
(0.258+0.258+6\cdot 0.81)^2
>
8(k_1^2+k_2^2).
\]
By \Cref{thm:lingam_reversal}, the
population JADE step recovers the wrong mixing matrix and hence the
reversed causal order \(X_2\to X_1\).

It is instructive to consider \(\varepsilon_i=\sigma Z_i\),
where \(Z_1,Z_2\) are i.i.d.\ standard Gaussian and independent of
\(\sigma\). After normalizing so that \(\E[\sigma^2]=1\), and assuming
\(\E[\sigma^4]<\infty\), we have \(k_1=k_2=3c\), regardless of the
distribution of \(\sigma\). Consequently, $(k_1+k_2+6c)^2 = 8(k_1^2+k_2^2)$,
so Gaussian scale mixtures lie on the boundary case of
\Cref{thm:lingam_reversal}. This
extends to all elliptical distributions.
\end{example}

In the two-dimensional setting, kurtosis-based FastICA also 
optimizes a marginal fourth-order contrast over orthogonal rotations.
Therefore it induces the same ordering of the candidate rotations as
the JADE objective, and under the same condition converges to the same
incorrect \(\pi/4\) rotation. After the usual LiNGAM post-processing
step, FastICA--LiNGAM returns the same incorrect causal order.

\subsection{Failure of residual-based LiNGAM}

The phenomenon is not confined to ICA-based implementations.
Residual-based LiNGAM procedures can also fail at the population level.
DirectLiNGAM \citep{Shimizu2011DirectLiNGAM} identifies a variable that
appears most nearly exogenous by checking whether its regression
residuals are independent of it. To illustrate the population behavior
under mean independence, consider the dependence score
\[
D(U,V):=\bigl|\Cov(U^2,V^2)\bigr|.
\]
If \(X_1\) is treated as the candidate exogenous variable, then
regressing \(X_2\) on \(X_1\) gives the residual \(R_{2\mid 1}=X_2-X_1=\varepsilon_2\),
so
\[
D(X_1,R_{2\mid 1})
=
\bigl|\Cov(\varepsilon_1^2,\varepsilon_2^2)\bigr|
=
c.
\]
If instead \(X_2\) is treated as the candidate exogenous variable, then
regressing \(X_1\) on \(X_2\) gives
\[
R_{1\mid 2}
=
X_1-\tfrac12 X_2
=
\tfrac12(\varepsilon_1-\varepsilon_2),
\qquad
X_2=\varepsilon_1+\varepsilon_2.
\]
A direct calculation yields
\(D(X_2,R_{1\mid 2})
=
\tfrac14\bigl|k_1+k_2-2c\bigr|\).
Therefore \(X_2\) appears more exogenous whenever
\(\tfrac14\bigl|k_1+k_2-2c\bigr|<c\). In particular, if
\(c>(k_1+k_2)/6\), the incorrect variable \(X_2\) has the smaller
dependence score.

Thus procedures that select the apparently most
exogenous variable by minimizing residual dependence can also fail under mean independence. The failure is again driven by
dependence in second moments, such as shared volatility, which is not
removed by regression. In the case \(k_1=k_2=k\), the
condition reduces to \(c>k/3\), the same threshold as in
\Cref{thm:lingam_reversal}.

Simulations in Section~\ref{sec:synthetic} suggest that this problematic behavior of LiNGAM algorithms under mean independence becomes more severe in higher dimensions. 

\section{Identification under order-dependent mean independence}
\label{sec:ident}

We formalize the assumptions behind LiMIAM. We first develop a
finite-moment identification argument based on the second and \(d\)-th
moment tensors, for a fixed integer \(d\ge 3\). We then turn to a distribution-level criterion, based on
residual mean independence.

\subsection{Mean independence and one-sided moment zeros}

Let \(U\) and \(V\) be random variables. We say that \(U\) is
\emph{mean independent} of \(V\) if
\[
\E[U\mid V]=\E[U]
\qquad \text{almost surely.}
\]
Mean independence is not symmetric. For example, if \(U\sim N(0,1)\)
and \(V=U^2\), then \(\E[U\mid V]=0=\E[U]\), but
\(\E[V\mid U]=U^2\neq \E[V]\).

In the sequel, all random variables are assumed mean centered: $\E[\varepsilon] = 0$. Our
results can also be expressed in terms of cumulants, but for simplicity
we work with moments.

Since the model is recursive, the variables can be arranged in any
compatible topological order. We write \(j\preceq i\) if \(j=i\) or if
\(j\) is an ancestor of \(i\) in the DAG. The graph-based first-moment
restriction used below is one-sided: if \(j\not\preceq i\), then
$\E(\varepsilon_j \mid \varepsilon_i) = 0$.
In a
fixed topological order, this implies the triangular condition obtained
by requiring later disturbances to be mean independent of earlier
ones.

\begin{proposition}[One-sided moment zeros]
\label{prop:one-sided-moment-zeros}
Let \(U\) and \(V\) be mean-zero random variables. If \(U\) is mean
independent of \(V\), then \(\E[UV^d]=0\) for all integers \(d\ge 0\)
whenever the moments exist. Conversely, if \(\E[UV^d]=0\) for all
\(d\ge 0\) and the moment generating function of \(V\) is finite in a
neighborhood of zero, then \(U\) is mean independent of \(V\).
\end{proposition}

\begin{proof}
If \(U\) is mean independent of \(V\), then
\[
\E[UV^d]
=
\E\!\left[V^d\E[U\mid V]\right]
=
\E\!\left[V^d\E[U]\right]
=0
\]
for all \(d\ge 0\). The converse follows from standard arguments using
analyticity of the moment generating function and density of
exponentials; see, for example,
\cite[Theorem~2.3]{RibotSeigalZwiernik2025BeyondICA}.
\end{proof}

\subsection{Finite-order LiMIAM}

We first study a finite-order version of LiMIAM, namely the
moment-level consequence of order-dependent mean independence at one
fixed order \(d\) given by \Cref{prop:one-sided-moment-zeros}.

\begin{assumption}[Finite-order LiMIAM]
\label{ass:weak-limiam}
The observed variables follow the recursive linear model
\eqref{eq:sem}--\eqref{eq:mix}. Equivalently, \(X=A\varepsilon\), where
\(A=(I-B)^{-1}\), and \(B\) is the weighted adjacency matrix of a DAG
\(\cG\). Let \(j\preceq i\) denote that \(j=i\) or that \(j\) is an
ancestor of \(i\) in \(\cG\). Assume that \(\E[\varepsilon]=0\),
\(\E[\varepsilon_i^2]>0\) for all \(i\), and
\(\E[\|\varepsilon\|^d]<\infty\) for some fixed integer \(d\ge 3\).
Moreover,
\[
\E[\varepsilon_i\varepsilon_j]=0
\qquad\text{for all } i\neq j,
\]
and
\[
\E[\varepsilon_i^{\,d-1}\varepsilon_j]=0
\qquad
\text{whenever } j\not\preceq i.
\]
\end{assumption}

\begin{remark}
Under full order-dependent mean independence, the corresponding
moment entries vanish at all orders. If these entries vanish up to order
\(d\), then the corresponding entries of the cumulant tensor
\(\kappa_d(\varepsilon)\) also vanish, by the moment--cumulant relation.
Thus one may often work with moments or cumulants. This is analogous to ICA:
instead of imposing full independence, which makes all mixed cumulants
vanish, one may impose a finite-order condition at one fixed order
\(d\ge 3\). In the setting of \Cref{ass:weak-limiam}, the relevant structure is not full
diagonality of the cumulant tensor, but the one-sided vanishing pattern
above.
It is simpler to state directly in terms of moments. 
\end{remark}

\subsection{Identification of source nodes from a pair of moments}

\Cref{ass:weak-limiam} implies a moment identity that singles out the
source nodes of the DAG. Related higher-order residual moment
identities were used by \citet{WangDrton2020HighDimensionalLiNGAM} for
high-dimensional LiNGAM under independent non-Gaussian disturbances; see their Theorem~1. In their setting, the vanishing of the
score follows from independence after adjustment for parents. Here the
same algebraic type of identity follows from one-sided moment restrictions
implied by order-dependent mean independence. The criterion is also
related in spirit to higher-order least squares \citep{Schultheiss2024}.

For \(i,j\in[p]\), define the score
\[
S_{ij}^{(d)}
=
\E[X_iX_j^{d-1}]\,\E[X_j^2]
-
\E[X_iX_j]\,\E[X_j^d].
\]

\begin{lemma}[Source identity]
\label{lem:mmi_source_ratio}
Under \Cref{ass:weak-limiam}, let \(j\) be a source node of \(\cG\).
Then \(S_{ij}^{(d)}=0\) for every \(i\in[p]\).
\end{lemma}

\begin{proof}
Since \(j\) is a source node, \(X_j=\varepsilon_j\). Also
\(X_i=\sum_k A_{ik}\varepsilon_k\). Hence
\[
\E[X_iX_j^{d-1}]
=
\sum_k A_{ik}\E[\varepsilon_k\varepsilon_j^{d-1}]
=
\sum_k A_{ik}\E[\varepsilon_j^{d-1}\varepsilon_k].
\]
By \Cref{ass:weak-limiam}, the moment
\(\E[\varepsilon_j^{d-1}\varepsilon_k]\) is zero whenever
\(k\not\preceq j\). Since \(j\) is a source, \(k\preceq j\) holds only
for \(k=j\). Thus only the term \(k=j\) remains, and
\[
\E[X_iX_j^{d-1}]
=
A_{ij}\E[\varepsilon_j^d]
=
A_{ij}\E[X_j^d].
\]
Similarly, using the diagonal covariance condition,
\[
\E[X_iX_j]
=
\sum_k A_{ik}\E[\varepsilon_k\varepsilon_j]
=
A_{ij}\E[\varepsilon_j^2]
=
A_{ij}\E[X_j^2].
\]
Combining the two identities gives \(S_{ij}^{(d)}=0\).
\end{proof}

\begin{lemma}[Generic failure for non-sources]
\label{lem:mmi_nonsource}
Under \Cref{ass:weak-limiam}, let \(j\) be a non-source node of
\(\cG\). Then, for generic admissible parameters, there exists
\(i\in[p]\) such that \(S_{ij}^{(d)}\neq 0\).
\end{lemma}

\begin{proof}
Choose a source ancestor \(i\) of \(j\), which exists because \(j\) is
not a source. Consider \(P_{ij}=S_{ij}^{(d)}\). We show that \(P_{ij}\)
is not the zero polynomial in the admissible parameters. Since
\(A_{jj}=1\), the moment \(\E[\varepsilon_j^d]\) appears in
\(\E[X_j^d]\) with coefficient one. It does not appear in
\(\E[X_iX_j^{d-1}]\), because \(X_i=\varepsilon_i\) and \(X_j^{d-1}\)
has total degree \(d-1\). It also does not appear in the second-order
moments. Hence \(P_{ij}\) is affine in \(\E[\varepsilon_j^d]\), with
coefficient \(-\E[X_iX_j]\). Since \(i\) is a source, \(X_i=\varepsilon_i\). Therefore, using the
diagonal covariance condition,
\[
\E[X_iX_j]
=
\E[\varepsilon_iX_j]
=
\sum_k A_{jk}\E[\varepsilon_i\varepsilon_k]
=
A_{ji}\E[\varepsilon_i^2].
\]
Now \(A_{ji}\) is not the zero polynomial in the structural coefficients:
indeed, \(A=(I-B)^{-1}=I+B+\cdots+B^{p-1}\), and \(A_{ji}\) is the sum
of products of edge coefficients over all directed paths from \(i\) to
\(j\). Since \(i\) is an ancestor of \(j\), at least one such path exists.
Thus \(A_{ji}\E[\varepsilon_i^2]\) is not identically zero. Therefore the
coefficient of \(\E[\varepsilon_j^d]\) in \(P_{ij}\) is not identically
zero, and \(P_{ij}\) is not the zero polynomial. Hence
\(S_{ij}^{(d)}\neq0\) outside a proper algebraic subset of the admissible
parameter space.
\end{proof}

\begin{theorem}[Generic identification from a pair of moments]
\label{thm:moment-identifiability}
Under \Cref{ass:weak-limiam}, the coefficient matrix \(B\) is
generically identifiable from \((\mu_2(X),\mu_d(X))\).
\end{theorem}

\begin{proof}
By \Cref{lem:mmi_source_ratio}, every source node \(j\) satisfies
\(S_{ij}^{(d)}=0\) for all \(i\in[p]\). By
\Cref{lem:mmi_nonsource}, every non-source node \(j\) violates this
identity for at least one \(i\), outside a proper algebraic subset of the
admissible parameter space. Hence the source nodes of \(\cG\) are
generically identifiable from \((\mu_2(X),\mu_d(X))\).

Choose one identified source node \(j\). Since \(X_j=\varepsilon_j\), the
proof of \Cref{lem:mmi_source_ratio} gives
\[
\frac{\E[X_iX_j]}{\E[X_j^2]}=A_{ij}
\qquad
\text{for every } i.
\]
Define residual variables \(X_i'=X_i-A_{ij}X_j\) for all \(i\neq j\).
Then
\[
X_i'
=
\sum_k A_{ik}\varepsilon_k-A_{ij}\varepsilon_j
=
\sum_{k\neq j} A_{ik}\varepsilon_k.
\]
Thus residualization removes the source disturbance \(\varepsilon_j\)
from all remaining variables. The residual vector has the same recursive
form for the reduced DAG obtained by deleting \(j\), and its mixing matrix
is the submatrix of \(A\) obtained by deleting the row and column
corresponding to \(j\). The finite-order LiMIAM moment conditions are also
inherited by the reduced disturbance vector.

The same argument can therefore be applied recursively. At each step, the
exceptional set is defined by the vanishing of finitely many nonzero
polynomials of the same type as in \Cref{lem:mmi_nonsource}. Since only
finitely many reduced DAGs can arise along source-removal recursions, the
union of these exceptional sets is contained in a proper algebraic subset
of the original admissible parameter space. Outside this set, the
recursion identifies a compatible topological order.

Once a topological order is found, \(B\) is recovered from
\(\mu_2(X)\) by successive linear regressions of each variable on its
predecessors in that order, equivalently by the corresponding
LDL decomposition. Hence \(B\) is generically identifiable from
\((\mu_2(X),\mu_d(X))\).
\end{proof}

\begin{remark}
\label{rem:single-moment-not-enough}
The coupled use of \(\mu_2(X)\) and \(\mu_d(X)\) is essential. A single
higher-order moment tensor does not identify the causal order; see
\Cref{sec:non-identifiability}.
\end{remark}

\begin{remark} \label{rem:limiam-subsumes-lingam}
The moment conditions in \Cref{ass:weak-limiam} are directional, but
\Cref{thm:moment-identifiability} remains valid if one imposes stronger
symmetric conditions. For example, it still holds if
\[
\E[\varepsilon_i\varepsilon_j]
=
\E[\varepsilon_i\varepsilon_j^{\,d-1}]
=
0
\qquad \text{for all } i\neq j.
\]
It also still holds if all off-diagonal entries of \(\kappa_d(\varepsilon)\) vanish.
In \Cref{fig:schematic}, LiMIAM can be applied to all three settings: those implied by independence, pairwise mean independence, and order-dependent mean independence.
The directional conditions are minimal for our argument; see \Cref{prop:worst-case-necessity}.
In Appendix~\ref{sec:genericity-binary} we give more details of these genericity conditions for the binary case.
\end{remark}

\subsection{Distribution-level source detection and residualization}
\label{sec:distribution-source-detection}

The finite-moment results above show that a pair of moments is already
enough for generic identification. We now turn to a distribution-level
criterion. There are two reasons for doing this. First, the primitive
assumption in LiMIAM is mean independence, not a finite collection of
moment equations. Second, the distribution-level formulation leads
directly to a residual source-detection criterion, closer in spirit to
DirectLiNGAM.

\begin{assumption}[LiMIAM]
\label{ass:limiam}
The observed variables follow the recursive linear model
\eqref{eq:sem}--\eqref{eq:mix}, where \(B\) is the weighted adjacency
matrix of a DAG \(\cG\), and \(\E[\varepsilon]=0\). Moreover, with
\(j\preceq i\) denoting ancestry in \(\cG\),
\[
\E[\varepsilon_j\mid\varepsilon_i]=0
\qquad
\text{whenever } j\not\preceq i.
\]
\end{assumption}

For \(i\neq j\), define the regression residual
\[
R_{i\mid j}
=
X_i-\frac{\E[X_iX_j]}{\E[X_j^2]}X_j.
\]
This is the same object that has a central role in DirectLiNGAM
\citep{Shimizu2011DirectLiNGAM}, where exogenous variables are
characterized through independence between a candidate regressor and the
residuals obtained after regressing the other variables on it. Here we
replace full independence by mean independence.

\begin{lemma}
\label{lem:residualid}
Under \Cref{ass:limiam}, let \(j\) be a source node of \(\cG\). Then
\(R_{i\mid j}\) is mean independent of \(X_j\) for all \(i\neq j\).
\end{lemma}

\begin{proof}
Since \(j\) is a source node, \(X_j=\varepsilon_j\). Write
\[
X_i=A_{ij}\varepsilon_j+\eta_i,
\qquad
\eta_i=\sum_{k\neq j}A_{ik}\varepsilon_k.
\]
By \Cref{ass:limiam}, \(\E[\varepsilon_k\mid\varepsilon_j]=0\) for
every \(k\neq j\), because \(k\not\preceq j\) when \(j\) is a source.
Hence
\[
\E[\eta_i\mid\varepsilon_j]
=
\sum_{k\neq j}A_{ik}\E[\varepsilon_k\mid\varepsilon_j]
=0.
\]
Since \(X_j=\varepsilon_j\), this gives \(\E[\eta_i\mid X_j]=0\). It
also implies \(\E[\eta_iX_j]=0\), so $\E[X_iX_j]
=
A_{ij}\E[X_j^2]$. Therefore \(R_{i\mid j}=X_i-A_{ij}X_j=\eta_i\), and
\(\E[R_{i\mid j}\mid X_j]=0\).
\end{proof}

The next result shows that, for a fixed pair of variables, residual mean
independence is equivalent to linearity of the corresponding conditional
mean.

\begin{proposition}
\label{prop:equivalence-mi-linearity}
Let \(X\) and \(Y\) be random variables such that \(\E[X]=\E[Y]=0\) and
\(\E[X^2]>0\). Let \(\beta=\E[XY]/\E[X^2]\). Then the following are
equivalent:
\begin{enumerate}[label=(\roman*)]
    \item the residual \(R_{Y\mid X}=Y-\beta X\) is mean independent of \(X\);
    \item \(\E[Y\mid X]\) is linear in \(X\), namely \(\E[Y\mid X]=\beta X\).
\end{enumerate}
If, in addition, the moment generating function of \(X\) is finite in a
neighborhood of the origin, then these conditions are also equivalent to
\begin{enumerate}[label=(\roman*),start=3]
    \item for every integer \(d\ge 1\),
    \[
    \E[YX^d]=\beta \E[X^{d+1}].
    \]
\end{enumerate}
\end{proposition}

\begin{proof}
Since \(\E[X]=\E[Y]=0\), the statement that \(\E[Y\mid X]\) is linear is
equivalent to \(\E[Y\mid X]=\beta X\) almost surely. By definition,
\[
\E[R_{Y\mid X}\mid X]
=
\E[Y\mid X]-\beta X.
\]
Hence \(R_{Y\mid X}\) is mean independent of \(X\) if and only if
\(\E[Y\mid X]=\beta X\). This proves the equivalence of (i) and (ii).

If \(\E[Y\mid X]=\beta X\), then for every integer \(d\ge 1\),
\[
\E[YX^d]
=
\E\!\left[X^d\E[Y\mid X]\right]
=
\beta\,\E[X^{d+1}],
\]
which gives (iii). Conversely, if (iii) holds and
\(h(X)=\E[Y\mid X]-\beta X\), then \(\E[h(X)X^d]=0\) for all
\(d\ge 0\). By the same argument as in
\Cref{prop:one-sided-moment-zeros}, this implies \(h(X)=0\) almost
surely, hence \(\E[Y\mid X]=\beta X\).
\end{proof}

A non-source variable can pass the residual criterion only if certain
conditional means are accidentally linear. We collect the needed
non-cancellation condition in the following assumption.

\begin{assumption}[Exogenous-variable detectability]
\label{ass:generic-limiam}
Assume \Cref{ass:limiam}. Moreover, in every recursively reduced model
obtained by successively removing source nodes and residualizing the
remaining variables, the variables \(X_j\) such that \(R_{i\mid j}\) is
mean independent of \(X_j\) for all \(i\neq j\) are exactly the source
nodes of the current reduced DAG.
\end{assumption}

\begin{remark}
    \Cref{ass:generic-limiam} is satisfied for general distributions, in the sense that the genericity conditions considered in \Cref{lem:mmi_nonsource} and \Cref{thm:moment-identifiability} are the finite-order consequences of \Cref{ass:generic-limiam}.
\end{remark}

\begin{theorem}[Generic identification from the distribution]
\label{thm:mmi_identification}
Under \Cref{ass:generic-limiam}, the coefficient matrix \(B\) is
identifiable from the joint distribution of \(X\).
\end{theorem}

\begin{proof}
The source nodes satisfy the residual mean independence criterion given in
\Cref{lem:residualid}. The converse holds in every recursively reduced
model by \Cref{ass:generic-limiam}. Hence the source nodes are
identifiable from the joint distribution. After choosing one source, we
residualize the remaining variables and apply the same argument
inductively. This recovers a compatible causal order. Once the order is
known, \(B\) is recovered by linear regression of each variable on its
predecessors.
\end{proof}

\begin{remark}
\label{rem:elliptical}
Gaussian and, more generally, elliptical distributions satisfy the
linearity condition in \Cref{prop:equivalence-mi-linearity}; see
\citet{kelker1970distribution}. Hence the direct residual criterion may
fail there even for non-exogenous variables. Non-Gaussian elliptical
disturbances are incompatible with LiNGAM for a different reason: in the
elliptical family, mutual independence occurs only in the Gaussian case.
\end{remark}

\subsection{Non-identifiability under weaker conditions}
\label{sec:non-identifiability}

The previous identification results use one-sided source-detection
conditions. Once a compatible source-removal order is fixed, these
conditions imply the triangular just-off-diagonal zero pattern in the
disturbance moments. We now explain why this triangular pattern cannot be
weakened arbitrarily. A natural question is whether some of these
one-sided zeros can be removed while retaining generic identifiability.
The next proposition shows that this is impossible in the worst case: if
one one-sided zero relation is dropped, then the same observed moments may
arise from two causal orders.

To formulate this, let \(\Sym^d(\R^p)\) denote the space of
\(p\times\cdots\times p\) symmetric tensors of order \(d\). Define
\[
\cV_d
=
\left\{
\cS\in \Sym^d(\R^p):
\cS_{i,\dots,i,j}=0 \text{ for all } i<j
\right\}.
\]
Let \(\mathrm{LU}(p)\) be the group of \(p\times p\) lower triangular
matrices with ones on the diagonal, acting on \(\Sym^d(\R^p)\) by
\[
[A \bullet \cS]_{i_1,\dots,i_d}
\coloneqq
\sum_{j_1,\dots,j_d=1}^p
A_{i_1,j_1}\cdots A_{i_d,j_d}\cS_{j_1,\dots,j_d}.
\]
For a fixed topological order, \Cref{ass:weak-limiam} implies
\(\mu_d(\varepsilon)\in\cV_d\), and the observed moment tensor is
\(\mu_d(X)=A\bullet\mu_d(\varepsilon)\).

\begin{proposition}[Worst-case necessity of the one-sided zero relations]
\label{prop:worst-case-necessity}
Fix \(d\ge 3\) and a total order \(1<\cdots<p\). For a set
\(W \subsetneq \{(i,j)\in [p]^2 : i<j\}\), let
\[
\cV_d(W)
=
\{\cS\in \Sym^d(\mathbb R^p) : \cS_{i,\dots,i,j}=0
\text{ for all } (i,j)\in W\}.
\]
Consider the coupled moment model
\[
(M,\T)=(A\bullet D,\;A\bullet \cS),
\qquad
A\in \mathrm{LU}(p),\quad D\in \cV_2,\quad \cS\in \cV_d(W).
\]
Then generic identifiability fails.
\end{proposition}

\begin{proof}
By relabeling, we may assume that the omitted relation is at the pair
\((1,2)\). It therefore suffices to consider the two-node case. Let
\[
M=
\begin{pmatrix}
m_{11} & m_{12}\\
m_{12} & m_{22}
\end{pmatrix}
\in \Sym^2(\mathbb R^2)
\]
with \(m_{11},m_{22}\neq 0\), and let \(\T\in \Sym^d(\mathbb R^2)\) be
arbitrary. For the orientation \(1\to 2\), set
\[
A=
\begin{pmatrix}
1&0\\
a&1
\end{pmatrix},
\qquad
a=\frac{m_{12}}{m_{11}},
\qquad
D=A^{-1}M(A^{-1})^\top.
\]
Then \(D\) is diagonal. The unique constraint
\(\cS_{1,\dots,1,2}=0\) has been removed, so every tensor
\(\cS=A^{-1}\bullet \T\) is admissible. Thus \((M,\T)\) belongs to the
model with order \(1<2\).

Now reverse the order. Let
\[
P=
\begin{pmatrix}
0&1\\
1&0
\end{pmatrix},
\qquad
\widetilde M=PMP^\top,
\qquad
\widetilde \T=P\bullet \T.
\]
Define
\[
\widetilde A=
\begin{pmatrix}
1&0\\
\widetilde a&1
\end{pmatrix},
\qquad
\widetilde a=\frac{\widetilde M_{12}}{\widetilde M_{11}}
=\frac{m_{12}}{m_{22}},
\qquad
\widetilde D=\widetilde A^{-1}\widetilde M(\widetilde A^{-1})^\top.
\]
Again \(\widetilde D\) is diagonal, and every tensor
\(\widetilde \cS=\widetilde A^{-1}\bullet \widetilde \T\) is admissible.
Undoing the permutation writes \((M,\T)\) via the reverse order \(2<1\).
\end{proof}

The next result shows that a single higher-order moment tensor is not
enough to identify the causal order, because the following decomposition
also exists for any permuted order.

\begin{theorem}[Higher-order LDL]
\label{thm:higher-order-LDL}
A generic \(\cT\in \Sym^d(\R^p)\) has a unique decomposition
\[
\cT=L\bullet \cD,
\]
where \(L\in \mathrm{LU}(p)\) and \(\cD\in \Sym^d(\R^p)\) has
\(\cD_{i,\dots,i,j}=0\) for all \(i<j\).
\end{theorem}

The proof is given in Appendix~\ref{sec:proofs-appendix}.
\Cref{thm:higher-order-LDL} explains why one cannot identify the causal
order from \(\mu_d(X)\) alone. The coupling with \(\mu_2(X)\) is what
breaks the symmetry, because \(\mu_2(X)\) and \(\mu_d(X)\) share the same
lower-triangular mixing matrix \(A\).
\section{Causal discovery algorithms under mean independence}
\label{sec:algorithm}

The identification results of~\Cref{sec:ident} suggest causal discovery algorithms based on measuring departure from mean independence. \Cref{thm:mmi_identification} ensures that under LiMIAM the source nodes can be identified recursively, and \Cref{lem:residualid} gives a constructive criterion.
The resulting algorithm mimics DirectLiNGAM \citep{Shimizu2011DirectLiNGAM}. At each iteration, one identifies a variable that appears most nearly
exogenous in the current reduced DAG by checking whether the residuals of
the other variables are mean independent of it, then regresses it out
and continues on the reduced system. The output is a compatible causal order.

\begin{algorithm}[t]
\caption{DirectLiMIAM}
\label{alg:DirectLiMIAM}
\begin{algorithmic}[1]
\Require i.i.d.\ samples $X^{(1)},\dots,X^{(N)} \in \R^p$, 
\Ensure compatible causal order $\pi_1 < \cdots < \pi_p$

\State $\cI \gets \{1,\dots,p\}$
\State $\pi \gets ()$ \Comment{start with an empty tuple}

\While{$\cI \neq \varnothing$}
\State $X_i \gets \operatorname{standardize}(X_i) \coloneqq \frac{X_i - \widehat{ \operatorname{Mean}} (X_i)}{\sqrt{\widehat \Var (X_i)}} \qquad \forall i\in\cI$
\State
\(
j_\star
\gets
\argmin_{j\in\cI}
\sum_{i\in\cI\setminus\{j\}}
\operatorname{MeanInd}(\widehat R_{i \mid j}, X_j), \qquad \text{with }\widehat R_{i \mid j} = X_i-\frac{\widehat{\Cov}(X_i,X_j)}{\widehat{\Var}(X_j)}X_j,
\)
\State
\(
X_i
\gets
R_{i \mid j_\star}
\qquad \forall i\in\cI\setminus\{j_\star\}
\)
\State $\pi \gets (\pi, j_\star)$
\Comment{append $j_\star$ to the end of $\pi$}
\State $\cI \gets \cI\setminus\{j_\star\}$
\EndWhile
\State \Return $\pi$
\end{algorithmic}
\end{algorithm}

Once we have a compatible causal order via \Cref{alg:DirectLiMIAM}, the adjacency matrix, and hence the DAG, can be estimated by conventional regression methods, such as ordinary least squares (OLS) or least absolute shrinkage and selection operator (LASSO), restricted to predecessors in the order.
We use OLS in Sections \ref{sec:synthetic} and \ref{sec:benchmarks}; our code available on GitHub also contains a LASSO implementation.
Determining which variable is next in the order requires a statistic that measures departure from mean independence, that is, from \(\E(R_{i\mid j}\mid X_j)=0\), which we denote by \(\operatorname{MeanInd}(R_{i\mid j},X_j)\). Several alternatives can be used.

\medskip
\noindent\textit{Kernel and nonparametric regression.}
One estimates the conditional mean nonparametrically using kernel or local polynomial regression and measures the magnitude of the fitted
conditional mean. For example, if \(\hat m(x)\approx \E[R_{i\mid j}\mid X_j=x]\), one may use
\begin{equation}\label{eq:kernellimiam} 
\operatorname{MeanInd}(\widehat R_{i\mid j},X_j)
=
\frac1N\sum_{\nu=1}^N (\hat m(X_j^{(\nu)}) - \bar R_{i\mid j} )^2
\end{equation}
as a dependence score; see for example \citet{Fan1996}. Here $\bar R_{i\mid j}$ denotes the average residual, i.e. the constant estimate.  

\medskip
\noindent\textit{Series methods.}
One approximates the conditional mean by a finite-dimensional basis
expansion,
\[
\widehat R_{i\mid j}
=
a_0+\sum_{k=1}^K a_k\phi_k(X_j)+u_i,
\]
and uses \eqref{eq:kernellimiam} to assess mean independence. Alternatively, one can measure an auxiliary \(R^2\) or the sum of squared fitted values as a
score. These are simple to implement and directly target conditional mean predictability; see \citet{Chen2007}.

\medskip
\noindent\textit{Moment-based methods.}
Mean independence implies
\[
\E\!\left[R_{i\mid j}g_k(X_j)\right]=0,
\qquad k=1,\dots,K,
\]
for any family of test functions \(g_k\). One may therefore form a
quadratic-form statistic based on these empirical moments, in the spirit
of the generalized method of moments.

All of these procedures test the weaker condition of mean independence rather than full independence. In contrast to the original DirectLiNGAM dependence measure, which detects arbitrary nonlinear dependence, the statistics above focus on whether the conditional mean of the residual varies with the candidate regressor.

\begin{remark}
The moment-based measures connect directly to our identification arguments. Using the definition of \(R_{i\mid j}\), the moment condition \(\E[R_{i\mid j}g(X_j)]=0\) becomes
\[
\frac{\E[X_ig(X_j)]}{\E[X_jg(X_j)]}
=
\frac{\Cov(X_i,X_j)}{\Var(X_j)}.
\]
For monomials \(g(x)=x^{d-1}\), this is the moment-ratio of~\Cref{sec:ident}.
\end{remark}

\section{Numerical experiments} 
\label{sec:synthetic}
 
We compare the finite-sample performance of DirectLiNGAM \cite{Shimizu2011DirectLiNGAM} and various versions of our DirectLiMIAM method: Kernel, Sieves and Moment based. We assess the absolute and relative performance of these methods under both independent and dependent errors.

\subsection{Simulation design}

For each of the 100 Monte Carlo replications, we generate a DAG with
\(p \in \{2,3,4,5,6\}\) nodes and sample size \(T=500\). The true
causal order is \(1,\dots,p\). Observations are simulated from the
structural model \eqref{eq:sem}, where \(B\) is strictly lower
triangular and its nonzero coefficients are drawn independently from a
uniform distribution on \([0.3,0.8]\).

The main ingredient of the simulation study is how the disturbances \(\varepsilon_i\) are generated. We first sample independent mean-zero auxiliary variables \(u_1,\dots,u_p\), and then transform them according to one of
several dependence designs. In the benchmark design, the disturbances
are simply \(\varepsilon_j=u_j\), so they are mutually independent. In
the four nonlinear designs below, the disturbances satisfy the
order-dependent mean-independence property
\[
\mathbb E\!\left[\varepsilon_j \mid \varepsilon_1,\dots,\varepsilon_{j-1}\right] = 0,
\]
while generally failing to be mutually independent.

We consider four distributions for an auxiliary variable
\(u\):
\begin{enumerate}
    \item uniform: \(u \sim U[-1,1]\);
    \item U-shaped Beta: \(u = 2V - 1\), where \(V \sim \mathrm{Beta}(0.5,0.5)\);
    \item concentrated Beta: \(u = 2V - 1\), where \(V \sim \mathrm{Beta}(2,2)\);
    \item Bimodal: $ u \sim
\frac{1}{2}\,  U[-1,-0.3]
+
\frac{1}{2}\,  U[0.3,1]$ 
\end{enumerate}
All four are symmetric around zero, but differ in shape. Similarly as in the ICA literature, stronger deviations from spherical distributions generally imply that it is easier to detect the correct causal order. We combine them with four dependence designs for the disturbances.

\medskip
\noindent\textbf{Independent design.}
This is the benchmark setting. We take
\[
\varepsilon_j = u_j,
\qquad j=1,\dots,p,
\]
so the disturbances are mutually independent.

\medskip
\noindent\textbf{Lagged heteroskedasticity design.}
The scale depends on a signed exponentially weighted history of the
past. We set $\varepsilon_1 = u_1$ and for \(j\ge 2\), define
\[
S_{j-1}=\sum_{\ell=1}^{j-1}\rho^{\,j-\ell-1}\varepsilon_\ell,
\qquad \rho\in(0,1),
\]
so that more recent disturbances receive larger weight. To make the
strength comparable across dimensions and across
auxiliary-variable distributions, we standardize this history: $\widetilde S_{j-1}
=
S_{j-1}/\operatorname{sd}(S_{j-1})$. We then set $\varepsilon_j = \sigma_j u_j$, 
$\sigma_j = \exp\!\left(\frac{\gamma}{2}\widetilde S_{j-1}\right)$,
where \(\gamma>0\) controls the strength of the heteroskedasticity.
The scale of the disturbance depends on the signed magnitude of all previous disturbances.

\medskip
\noindent\textbf{Threshold design.}
The scale changes discontinuously depending on the sign of the
average past disturbance. We set $\varepsilon_1 = u_1$, and for \(j\ge 2\),
\[
\varepsilon_j =
\begin{cases}
u_j, & \bar\varepsilon_{j-1}\le 0,\\[1ex]
2u_j, & \bar\varepsilon_{j-1}>0,
\end{cases}
\qquad
\bar\varepsilon_{j-1}
=
\frac{1}{j-1}\sum_{\ell=1}^{j-1}\varepsilon_\ell.
\]
The conditional variance is larger when the average of the
previous disturbances is positive.

\medskip
\noindent\textbf{Conditional mixture design.}
The distribution of the current disturbance depends on the past
through a two-component mixture. Let
\[
\bar\varepsilon_{j-1}
=
\frac{1}{j-1}\sum_{\ell=1}^{j-1}\varepsilon_\ell,
\qquad
p_{j-1}
=
\frac{1}{1+\exp(-2\bar\varepsilon_{j-1})}.
\]
For \(j\ge 2\), define
\[
\varepsilon_j =
\begin{cases}
u_j^{(L)}, & \text{with probability } p_{j-1},\\[1ex]
2.5\,u_j^{(H)}, & \text{with probability } 1-p_{j-1},
\end{cases}
\]
where \(u_j^{(L)}\) and \(u_j^{(H)}\) are independent draws from the
chosen auxiliary-variable distribution. Since both mixture components
are centered at zero, the conditional mean remains zero, but the
conditional distribution and tail behavior vary with the disturbance
history.

Combining the four dependence designs with the four auxiliary-variable
distributions yields \(16\) simulation environments.

\subsection{Implementation details} 

For DirectLiNGAM \cite{Shimizu2011DirectLiNGAM} we use the pairwise likelihood score of \cite{Hyvarinen2013} to measure departures from independence. For the mean independent kernel implementation we rely on a local linear regression to compute the regression function where the bandwidth is chosen using $K$-fold cross-validation with $K=5$. Using the kernel estimate we compute \eqref{eq:kernellimiam} as our measure for mean independence. For the sieves implementation we use a cubic B-spline basis with equi-spaced knots and a number of splines that is selected using $5$-fold cross-validation. Finally, for the moment-based implementation we use $g(x)=(x^2, x^3)'$ and measure deviations from mean independence from the inner-product of the sample moments. 

We report exact DAG recovery and the structural Hamming distance (SHD) between the estimated and true DAGs. The SHD between two DAGs \(\cG\) and \(\widetilde{\cG}\) is the number of directed edges present in \(\cG\) but not \(\widetilde{\cG}\), plus the number present in \(\widetilde{\cG}\) but not \(\cG\).

\subsection{Results} 

Figures \ref{fig:synthetic-comparison1}-\ref{fig:synthetic-comparison2} document the main results for the causal orderings. The order success rate and SHD show qualitatively comparable results. First, for the independent design, there are no systematic differences between the independent and mean independent algorithms. This is reassuring for the LiMIAM methods, as there do not appear to be large losses when the true errors are independent. 

In the other three designs, which violate mutual independence, it is clear that the conventional DirectLiNGAM procedure performs worse. Most notably, the probability of getting the ordering correct drops severely for many designs and the SHD increases. In contrast, the DirectLiMIAM algorithm continues to perform adequately. Among the implementations, the kernel implementation appears most reliable. It is never the worst and often the best.

\begin{remark}
    We also compared our methods with ICA-LiNGAM \cite{Shimizu2006LiNGAM} and NOTEARS \cite{zheng2018dags}. Their performance was found to be inferior to that of the DirectLiNGAM method and, therefore, omitted from the discussion.  
\end{remark}

\begin{figure}[t]
    \centering
    \includegraphics[width=1\linewidth]{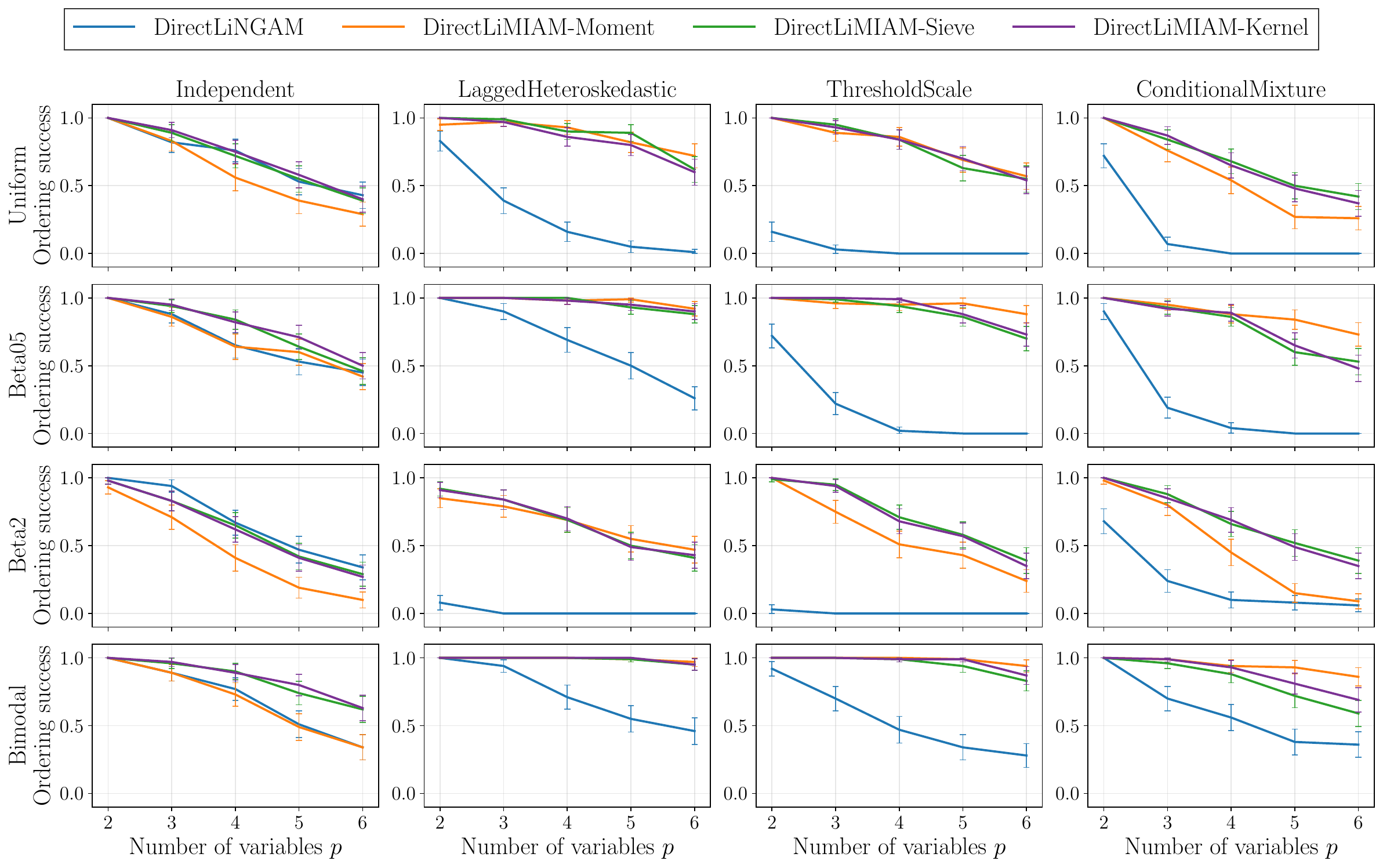}
    \caption{Exact ordering success rate across causal recovery algorithms.}
    \label{fig:synthetic-comparison1}
\end{figure}

\begin{figure}[t]
    \centering
    \includegraphics[width=1\linewidth]{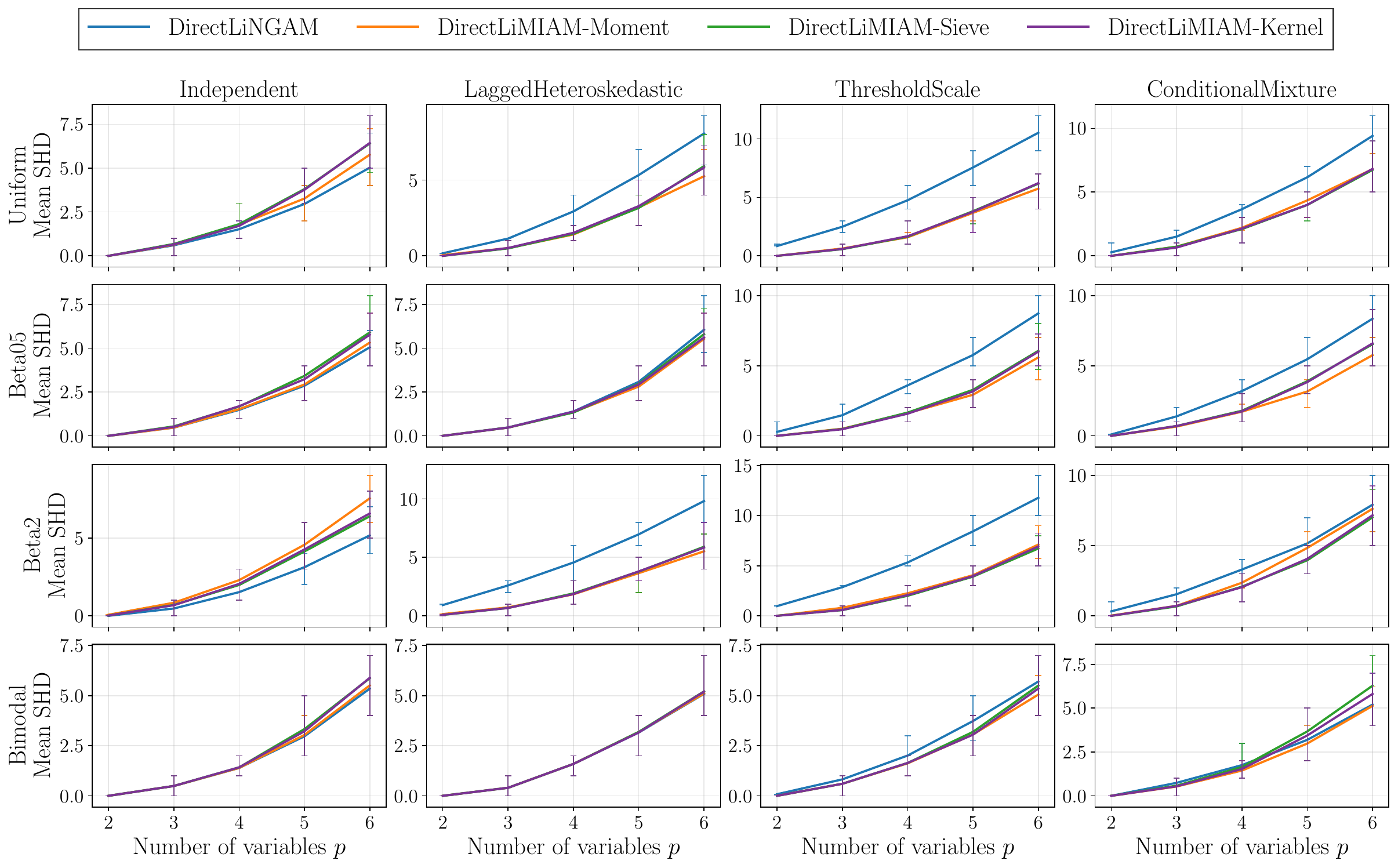}
    \caption{Structural Hamming distance across causal recovery algorithms.  }
    \label{fig:synthetic-comparison2}
\end{figure}

\section{Causal discovery for the oil market} 
\label{sec:benchmarks}

Since the seminal work of \cite{Sims1980}, causal inference in macroeconomics often proceeds in two steps. First, an essentially unrestricted vector autoregressive model is used to describe the dynamics of a set of macroeconomic or financial variables. Second, contemporaneous identification restrictions are imposed to trace out the dynamic causal effects of the exogenous inputs. The SVAR model with $k$ lags is described by 
\begin{equation}
    X_t = c + \Phi_1 X_{t-1} + \ldots + \Phi_k X_{t-k} + A \varepsilon_t ~, \qquad t=1, \ldots, T ~, 
\end{equation}
where the disturbances $\varepsilon_t = (\varepsilon_{1t}, \ldots , \varepsilon_{pt})'$ often have an economic interpretation. This model faces a similar identification problem as the static LSEM \eqref{eq:mix}, which can be resolved by imposing restrictions on $A$ and on the distribution of the disturbances $\varepsilon_t$.  

An early identification approach is based on timing restrictions, which set $A$ to be lower triangular according to an order defined a priori by economic theory or institutional knowledge \citep{Sims1980}. More recently, works such as \citep{LanneLutkepohl2010,Hyvarinen2010,Lanne2017,GourierouxMonfortRenne2017,Hoesch2024} have used the ICA assumption, that the components of the disturbances $\varepsilon_i$ are mutually independent, to recover $A$ up to sign and permutation. Additional identifying restrictions are discussed in the review by \cite{Ramey2016} and in the textbook of \cite{KilianLutkepohl.17}.   

Our framework also imposes that $A$ is lower triangular, after a permutation, but does not require knowledge of the permutation; i.e., of the causal ordering. The only requirement is that the disturbances satisfy the order-dependent mean-independence \Cref{ass:limiam}. Mean independence allows for common volatility changes, which are omnipresent in macroeconomic time series.  

In this empirical study we explore our identification results and the DirectLiMIAM algorithm for a monthly SVAR model of the oil market. Such models have previously been studied in \cite{Kilian2009} and \cite{Kanzig2021}, among many others. We refer to \cite{Kilian2023} for a recent review of oil market modeling.   

\subsection{Data} 

We consider an SVAR model for five variables: WTI oil price (POIL), world oil production (OILPROD), world oil inventories (OILSTOCKS), U.S. industrial production (IP), and the U.S. consumer price index (CPI). The data are monthly from January 1975 until June 2025. We include $k=24$ lags in the SVAR to ensure that all serial correlation is captured and that no serial correlation remains in the errors. 

To this baseline oil SVAR we add a sixth variable that measures surprises in Organization of the Petroleum Exporting Countries (OPEC) announcements from high-frequency changes in oil price futures (SURPRISE). \cite{Kanzig2021} argues that these surprises should come first in the causal ordering of the SVAR disturbances and presents exercises to verify this claim. In our study we use this variable to capture part of the ``ground truth'' of the causal ordering and thus to assess the validity of our methodology.

\subsection{Independence tests} 

We start by testing whether the disturbance independence assumption underlying DirectLiNGAM is reasonable for oil market disturbances. For this, we first estimate the coefficients $c,\Phi_1, \ldots, \Phi_k$ using least squares, compute the residuals of the SVAR, standardize them so that they have unit variance, and apply the DirectLiNGAM algorithm. Using the estimated $A = (I - B)^{-1}$, we can estimate the disturbances:
\[
\hat \varepsilon_t = \hat A^{-1} \hat U_t~, \qquad t=1, \ldots, T
\]
where the residuals are
\begin{equation}\label{eq:varres}
\hat U_t = \operatorname{standardize}\left( X_t - \hat c - \hat \Phi_1 X_{t-1} - \ldots - \hat \Phi_k X_{t-k} \right)~.  
\end{equation}
We apply the mutual independence test of \cite{Pfister2018} to $\{ \hat \varepsilon_t \}$ to assess whether the independence assumption underlying DirectLiNGAM is plausible. 

The $p$-value of the test is 0.024, rejecting the independence assumption. Moreover, the causal ordering produced by the DirectLiNGAM algorithm (left panel of \Cref{fig:oilmarketDAGs}) does not correspond to economic intuition. First, OPEC surprises are not ordered first in the DAG, and second, oil production disturbance is last in the ordering, whereas we would expect production to come earlier, since oil production adjusts with delay to economic developments. Oil production cannot scale up quickly, making an immediate dependence on economic conditions and oil prices implausible \citep{Kilian2009}. In summary, the conventional DirectLiNGAM approach does not give interpretable results when applied to the oil market SVAR.  

\begin{figure}[t]
    \centering
    \includegraphics[width=0.75\linewidth]{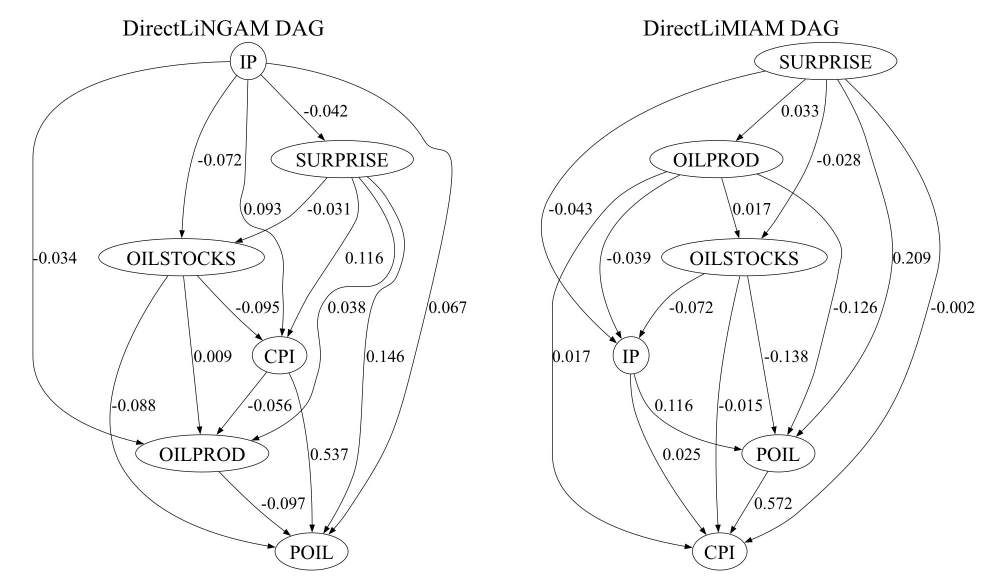} 
    \caption{Estimated oil market DAGs for LiNGAM and LiMIAM.}
    \label{fig:oilmarketDAGs}
\end{figure}

\subsection{DirectLiMIAM results}  

Next, we apply our DirectLiMIAM algorithm, using the local linear kernel implementation, to the SVAR residuals. The DAG is shown in the right panel of \Cref{fig:oilmarketDAGs}. The estimated DAG now orders the surprises in the OPEC announcements first, consistent with \cite{Kanzig2021}, who argues that these announcements provide exogenous changes in oil supply and should not be a function of other contemporaneous inputs. Further corroborating evidence is the positive link from the OPEC surprise to world oil production, which subsequently reduces the price of oil and increases the U.S. general price level (CPI). Also, the surprise increases the price of oil and reduces inventories.  

The estimated $B$ matrix, with bootstrap standard errors in the subscripts, becomes 
\medskip 
\begin{center}
\begin{tabular}{lccccc}
 & SURPRISE & OILPROD & OILSTOCKS & IP & POIL \\
OILPROD & $\phantom{-}0.033_{(0.043)}$ & -- &  &   &      \\
OILSTOCKS & $-0.028_{(0.043)}$ & $\phantom{-}0.017_{(0.042)}$ & -- &  &     \\
IP & $-0.043_{(0.076)}$ & $-0.039_{(0.055)}$ & $-0.072_{(0.062)}$ & -- &     \\
POIL & $\phantom{-}0.209_{(0.046)}$ & $-0.126_{(0.041)}$ & $-0.138_{(0.034)}$ & $\phantom{-}0.116_{(0.055)}$ & --   \\
CPI & $-0.002_{(0.038)}$ & $\phantom{-}0.017_{(0.029)}$ & $-0.015_{(0.032)}$ & $\phantom{-}0.025_{(0.052)}$ & $\phantom{-}0.572_{(0.047)}$ \\
\end{tabular}
\end{center}
\medskip 
There is substantial uncertainty in the estimated coefficients. Nevertheless, the OPEC surprises have a significant effect on oil prices, which in turn significantly affect U.S. inflation. 

\subsection{Mean independence tests}  

Next, we assess the validity of the order-dependent mean independence assumption. Specifically, we conduct a specification test for 
\[
H_{0}:
\mathbb E \left[
\varepsilon_{it}
\mid
 \varepsilon_{1t},\ldots, \varepsilon_{i-1,t}
\right]
=
0 \quad \forall \quad i \geq 2~.
\]
We build our test statistics using a kernel regression approach. That is, we recover the SVAR residuals $\hat U_t$, run the DirectLiMIAM algorithm to get $\hat A = (I - \hat B)^{-1}$ and $\hat \varepsilon_t = \hat A^{-1} \hat U_t$  for $t=1,\ldots,T$. We then reorder the disturbances according to the estimated causal ordering. Let $\tilde{\bm \varepsilon}_{i}$ denote the $T \times 1$ vector of centered realizations of \(\hat\varepsilon_{it}\), and let \(\widetilde{ K}_{i-1}\) denote a centered $T \times T$ Gaussian kernel matrix constructed from \((\hat \varepsilon_{1,t},\ldots,\hat \varepsilon_{i-1,t})\). The test statistic is
\[
\hat T
=
\sum_{i=2}^p \hat T_i, \qquad \text{with} \qquad \hat T_i
=
\tilde{\bm \varepsilon}'_{i} \,
\widetilde{ K}_{i-1}
\, \tilde{\bm \varepsilon}_{i}.
\]
Under \(H_{0}\), the conditional mean of \(\hat\varepsilon_{it}\) is zero, so \(\hat T_i\) should be close to zero. We obtain \(p\)-values by permutation, randomly permuting \(\hat\varepsilon_{it}\) while holding the conditioning variables fixed. 
When applying this test to the oil market residuals, we find a $p$-value of 0.772, and thus do not reject the ordered mean-independence hypothesis.  

\section{Conclusion} 

We have introduced LiMIAM, a model and algorithm for causal discovery under mean independence and linearity. We demonstrated its identifiability and provided a practical algorithm, called DirectLiMIAM, that is evaluated both in theory and in synthetic experiments. We showed that LiMIAM works under any notion of independence at least as strong as our order-dependent mean-independence (or its implication for a fixed $d$-th order moment tensor). If \(\varepsilon_i\) is mean independent of \( \varepsilon_j \) for all $i \neq j$, then the unrestricted mixing matrix in \(X=A\varepsilon\) is generically identifiable up to signed permutation of its columns by \citet{RibotSeigalZwiernik2025BeyondICA}. If \(A=(I-B)^{-1}\), the LiNGAM permutation-and-scaling step identifies \(B\). The present paper goes further by removing the need for full pairwise mean independence, using the ordered structure of the model, and developing a DirectLiNGAM analogue that does not rely on ICA. In our experiments, the pairwise mean independent component analysis algorithm of \citet{RibotSeigalZwiernik2025BeyondICA} was not competitive with DirectLiMIAM, even when full pairwise mean independence held.

We conclude by commenting on a direction for future work regarding the linearity of the causal dependencies.
Given any two random variables \(X\) and \(Y\), we can always write
\[
X = \E[X \mid Y] + \varepsilon
\]
such that \(\varepsilon\) is mean independent of \(Y\). Hence
identifiability cannot hold if one allows arbitrary nonlinear causal
dependencies. The distribution-level route in \Cref{sec:distribution-source-detection}
relies on the fact that, for linear regression, mean independence of the
residual \(R_{Y\mid X}\) is equivalent to linearity of the conditional
mean \(\E[Y\mid X]\); see \Cref{prop:equivalence-mi-linearity}. One could replace the linear family by a richer class
\(\mathcal F\) of functions and define the residual using the best
\(\mathcal F\)-approximation of \(Y\) by a function of \(X\). An analogue of \Cref{prop:equivalence-mi-linearity} would
characterize when the residual is mean independent of \(X\) in terms of
whether \(\E[Y\mid X]\) belongs to \(\mathcal F\).

This suggests that, if one allows nonlinear structural
relations, identifiability can still be pursued relative to
the chosen function class \(\mathcal F\). The corresponding genericity
assumption would exclude observational laws for which a
reverse-direction representation also belongs to \(\mathcal F\). The role here of linearity and
\Cref{ass:generic-limiam} would be replaced by a nonlinear function
class with analogous exogenous-variable detectability
assumption adapted to that class.

\bibliographystyle{plainnat}
\bibliography{bibliography}

\markboth{APPENDIX}{APPENDIX}
\appendix

\section{Omitted proofs} \label{sec:proofs-appendix}

\begin{proof}[Proof of \Cref{thm:lingam_reversal}]
The ICA step in LiNGAM first whitens the data. We introduce whitening
to analyze the behavior of ICA-based LiNGAM contrasts. It is
not a natural step for our proposed method, because whitening replaces
the triangular mixing matrix \(A=(I-B)^{-1}\) by an orthogonal factor
and therefore discards the lower triangular structure induced by the DAG.

Since \(\Cov(X)=AA^\top\), the whitened vector is
\[
Y_0=(AA^\top)^{-1/2}X=U_0\varepsilon,
\qquad
U_0:=(AA^\top)^{-1/2}A,
\]
where \(U_0\) is orthogonal. ICA then applies another orthogonal matrix
\(Q\) to the whitened data, so the recovered components are
\[
\hat\varepsilon=QY_0=QU_0\varepsilon.
\]
Because \(QU_0\) is orthogonal, and every \(2\times 2\) orthogonal
matrix is a rotation or a reflection, while reflections are immaterial
for the present argument, we may write
\[
\hat\varepsilon(\theta)=R(\theta)\varepsilon,
\qquad
R(\theta)=
\begin{pmatrix}
\cos\theta & -\sin\theta\\
\sin\theta & \cos\theta
\end{pmatrix}.
\]
The value \(\theta=0\) corresponds to the true source coordinates, while
\(\theta=\pi/4\) gives
\[
\hat\varepsilon(\pi/4)
=
\frac{1}{\sqrt2}
\begin{pmatrix}
1 & -1\\
1 & 1
\end{pmatrix}
\varepsilon.
\]
In the present two-dimensional setting, the population JADE objective is
\(g(\theta)=\kappa_4(\hat\varepsilon_1(\theta))^2+
\kappa_4(\hat\varepsilon_2(\theta))^2\). By multilinearity of cumulants
and the assumptions
\(\kappa_4(\varepsilon)_{1112}=\kappa_4(\varepsilon)_{1222}=0\),
\[
\kappa_4(\hat\varepsilon_1(\theta))
=
k_1\cos^4\theta+k_2\sin^4\theta+6c\cos^2\theta\sin^2\theta,
\]
\[
\kappa_4(\hat\varepsilon_2(\theta))
=
k_1\sin^4\theta+k_2\cos^4\theta+6c\cos^2\theta\sin^2\theta.
\]
Using
\[
\sin^4\theta+\cos^4\theta=1-2\sin^2\theta\cos^2\theta,
\qquad
\cos^4\theta-\sin^4\theta=\cos(2\theta),
\]
and
\[
\sin^2\theta\cos^2\theta=\frac{\sin^2(2\theta)}{4}
=\frac{1-\cos^2(2\theta)}{4},
\]
we can rewrite these as
\[
\kappa_4(\hat\varepsilon_1(\theta))
=
a+b\cos(2\theta)-q\cos^2(2\theta),
\]
\[
\kappa_4(\hat\varepsilon_2(\theta))
=
a-b\cos(2\theta)-q\cos^2(2\theta),
\]
where
\[
a=\frac{k_1+k_2+6c}{4},
\qquad
b=\frac{k_1-k_2}{2},
\qquad
q=\frac{6c-(k_1+k_2)}{4}.
\]
Setting \(u=\cos^2(2\theta)\in[0,1]\), the linear terms in
\(\cos(2\theta)\) cancel when the two squares are added, and so
\[
g(\theta)=2(a-qu)^2+2b^2u.
\]
Thus \(g\) is a quadratic polynomial in \(u\) with nonnegative leading
coefficient \(2q^2\). Its maximum over \(u\in[0,1]\) is therefore
attained at the boundary points \(u\in\{0,1\}\), corresponding to
\(\theta\in\{0,\pi/4\}\pmod{\pi/2}\).

At \(\theta=0\), we have \(g(0)=k_1^2+k_2^2\). At
\(\theta=\pi/4\), both recovered components have fourth cumulant
\((k_1+k_2+6c)/4\), so
\(g(\pi/4)=(k_1+k_2+6c)^2/8\). Comparing these two values gives the
stated inequality criterion.

It remains to compute the recovered mixing matrix in the case
\(\theta=\pi/4\). Since \(X=A\varepsilon=\hat A\hat\varepsilon(\pi/4)\),
we obtain
\[
\hat A
=
A\frac{1}{\sqrt2}
\begin{pmatrix}
1 & 1\\
-1 & 1
\end{pmatrix}
=
\frac{1}{\sqrt2}
\begin{pmatrix}
1 & 1\\
0 & 2
\end{pmatrix}.
\]
A direct calculation then gives
\[
\hat W=\hat A^{-1}
=
\frac{1}{\sqrt2}
\begin{pmatrix}
2 & -1\\
0 & 1
\end{pmatrix}.
\]
LiNGAM rescales the rows of \(\hat W\) to make the diagonal equal to
one, yielding
\[
\hat W_{\mathrm{scaled}}
=
\begin{pmatrix}
1 & -1/2\\
0 & 1
\end{pmatrix}\quad\mbox{and so }\qquad \hat B=I-\hat W_{\mathrm{scaled}}
=
\begin{pmatrix}
0 & 1/2\\
0 & 0
\end{pmatrix}.
\]
This is exactly the structural matrix associated with the reversed
causal order \(X_2\to X_1\).
\end{proof}

\begin{proof}[Proof of \Cref{thm:higher-order-LDL}]
We argue by induction on \(p\). The case \(p=1\) is trivial.
Assume the statement holds for \(p-1\), and let \(\cT\in \Sym^d(\R^p)\)
be generic. We construct \(M \in \mathrm{LU}(p)\) such that
\[
[M\bullet \cT]_{i,\dots,i,j}=0
\qquad \text{for all } 1\le i<j\le p.
\]
Then \(\cD=M\bullet \cT\) and \(L=M^{-1}\). Write \(M=M^{(2)}M^{(1)}\),
where
\[
M^{(1)}=
\begin{pmatrix}
1 & 0 & 0 & \cdots & 0\\
M^{(1)}_{21} & 1 & 0 & \cdots & 0\\
M^{(1)}_{31} & 0 & 1 & \cdots & 0\\
\vdots & \vdots & \vdots & \ddots & \vdots\\
M^{(1)}_{p1} & 0 & 0 & \cdots & 1
\end{pmatrix}
\quad\mbox{and}\qquad
M^{(2)}=
\begin{pmatrix}
1 & 0 & \cdots & 0\\
0 & & & \\
\vdots & & \widetilde M & \\
0 & & &
\end{pmatrix},
\]
with \(\widetilde M \in LU(p-1)\). For each \(j\ge 2\),
\[
[M^{(1)}\bullet \cT]_{1,\dots,1,j}
=
\sum_{k=1}^j M^{(1)}_{jk}\,\cT_{1,\dots,1,k}
=
M^{(1)}_{j1}\,\cT_{1,\dots,1}+\cT_{1,\dots,1,j}.
\]
Hence imposing
\[
[M^{(1)}\bullet \cT]_{1,\dots,1,j}=0
\qquad \text{for all } j\ge 2
\]
uniquely determines
\[
M^{(1)}_{j1}
=
-\frac{\cT_{1,\dots,1,j}}{\cT_{1,\dots,1}},
\qquad j=2,\dots,p,
\]
provided \(\cT_{1,\dots,1}\neq 0\), which holds generically. Let
\(\widetilde{\cT}=M^{(1)}\bullet \cT\). By construction,
\[
\widetilde{\cT}_{1,\dots,1,j}=0
\qquad \text{for all } j\ge 2.
\]
Now let \(1<i<j\le p\). Since \(M^{(2)}_{k1}=0\) for all \(k>1\), we
obtain
\begin{align*}
[M^{(2)}\bullet \widetilde{\cT}]_{i,\dots,i,j}
&=
\sum_{k_1=1}^i \cdots \sum_{k_{d-1}=1}^i \sum_{k_d=1}^j
M^{(2)}_{ik_1}\cdots M^{(2)}_{ik_{d-1}}M^{(2)}_{jk_d}
\,\widetilde{\cT}_{k_1,\dots,k_d}\\
&=
\sum_{k_1=2}^i \cdots \sum_{k_{d-1}=2}^i \sum_{k_d=2}^j
M^{(2)}_{ik_1}\cdots M^{(2)}_{ik_{d-1}}M^{(2)}_{jk_d}
\,\widetilde{\cT}_{k_1,\dots,k_d}.
\end{align*}
These equations involve only the \((p-1)\)-dimensional subtensor of
\(\widetilde{\cT}\) indexed by \(\{2,\dots,p\}\) and the lower unit
triangular block \(\widetilde M\). Since a generic subtensor remains
generic, the induction hypothesis applies and uniquely determines
\(\widetilde M\) by imposing
\[
[M^{(2)}\bullet \widetilde{\cT}]_{i,\dots,i,j}=0
\qquad \text{for all } 1<i<j\le p.
\]
Moreover,
\[
[M^{(2)}\bullet \widetilde{\cT}]_{1,\dots,1,j}=0
\qquad \text{for all } j\ge 2, 
\]
because \(\widetilde{\cT}_{1,\dots,1,j}=0\) and \(M^{(2)}_{j1}=0\) for
all \(j\ge 2\). This proves existence and uniqueness.
\end{proof}

\section{Genericity conditions in the binary case}\label{sec:genericity-binary}

Consider the finite order LiMIAM setup from \Cref{ass:weak-limiam} for the two-variable DAG $X_1 \to X_2$. That is, we have the LSEM
\[
\begin{cases}
    X_1 = \varepsilon_1 \\
    X_2 = B_{21} X_1 + \varepsilon_2
\end{cases}
\]
such that $B_{21} \neq 0$, $\E[\varepsilon] = 0$, $\E[\varepsilon_1^2], \E[\varepsilon_2^2] > 0$, $\E[\varepsilon_1 \varepsilon_2] =0$, $\E[\|\varepsilon\|^d] < \infty$, and $\E[\varepsilon_1^{d-1} \varepsilon_2] = 0$ for some integer $d \geq 3$. The genericity conditions of \Cref{lem:mmi_nonsource} are given by
\[
S^{(d)}_{12} \coloneqq \E[X_1X_2^{d-1}] \E[X_2^2] - \E[X_1X_2] \E[X_2^d] \neq 0.
\]
Expanding this expression, we obtain
\[
\begin{aligned}
S^{(d)}_{12}
&= \E\!\left[\varepsilon_1(B_{21}\varepsilon_1+\varepsilon_2)^{d-1}\right]\E\!\left[(B_{21}\varepsilon_1+\varepsilon_2)^2\right]
-\E\!\left[\varepsilon_1(B_{21}\varepsilon_1+\varepsilon_2)\right]\E\!\left[(B_{21}\varepsilon_1+\varepsilon_2)^d\right] \\[0.5em]
&=
-B_{21}\E[\varepsilon_1^2]\E[\varepsilon_2^d]
+\sum_{\substack{k=1\\ k\neq d-1}}^{d}
\left[
\binom{d-1}{k-1}B_{21}^{k-1}\E[\varepsilon_2^2]
-\binom{d-1}{k}B_{21}^{k+1}\E[\varepsilon_1^2]
\right]
\E[\varepsilon_1^k\varepsilon_2^{\,d-k}],
\end{aligned}
\]
which is not identically zero. Under pairwise mean independence (i.e., if we also assume $\E[\varepsilon_1 \mid \varepsilon_2] = 0$), the term $k=1$ in the expression above would also vanish, but we would still get a nonzero polynomial. Under independence and using cumulants instead of moments, all off-diagonal entries vanish, so the genericity condition becomes
\[
-B_{21}\kappa_2(\varepsilon)_{11} \kappa_d(\varepsilon)_{2, \dots, 2}
+B_{21}^{d-1} \kappa_2(\varepsilon)_{22} \kappa_d(\varepsilon)_{1, \dots, 1} \neq 0.
\]
This explains why LiMIAM subsumes LiNGAM, cf. \Cref{fig:schematic}.

\end{document}